\documentclass[conference]{IEEEtran}
\IEEEoverridecommandlockouts
\usepackage{cite}
\usepackage{graphicx} 
\usepackage{amsmath}
\usepackage{graphicx}
\usepackage{algorithm}
\usepackage{algpseudocode} 
\usepackage{algorithmicx}  
\usepackage{url}
\usepackage{booktabs} 
\usepackage{subcaption}
\usepackage{amsmath}
\usepackage{amssymb}
\usepackage{amsthm}

\newtheorem{lemma}{Lemma}

\def\BibTeX{{\rm B\kern-.05em{\sc i\kern-.025em b}\kern-.08em
    T\kern-.1667em\lower.7ex\hbox{E}\kern-.125emX}}
\begin{document}

\title{Design and Analysis of Chirp-Layered Superposition Coding for LoRa\\
}

\author{\IEEEauthorblockN{Jingxiang Huang}
\IEEEauthorblockA{\textit{Faculty of Computer Science} \\
\textit{Dalhousie University}\\
Halifax, Canada \\
jn661503@dal.ca}
\and
\IEEEauthorblockN{Samer Lahoud}
\IEEEauthorblockA{\textit{Faculty of Computer Science} \\
\textit{Dalhousie University}\\
Halifax, Canada \\
sml@dal.ca}
}

\maketitle

\begin{abstract}
This paper investigates the design of chirp-layered superposition coding for LoRa, where an additional waveform is linearly superposed on a standard LoRa transmission with minimal impact on the LoRa demodulation process. We first show that any non-zero superposed signal perturbs the output of the standard dechirp-and-DFT demodulator, and then characterize the class of superposed waveforms that minimize this degradation under a given power budget. In particular, we show that a high spreading factor (high-SF) LoRa waveform superposed on a low-SF signal (e.g., SF12 on SF7) can be designed so that its impact on the standard LoRa demodulation remains small. As a result, within each low-SF symbol interval, the high-SF segment can be treated as a quasi-narrowband carrier that conveys an additional BPSK stream. We derive analytical error-rate expressions for both the low-SF LoRa layer and the superposed high-SF layer, and validate them through Monte Carlo simulations. The proposed chirp-layered superposition coding scheme improves the spectral efficiency of LoRa-based links and uses a relatively simple transceiver architecture.

\end{abstract}


\section{Introduction}

\subsection{Background}

LoRa has become one of the most widely used physical-layer technologies for low-power wide-area networks (LPWANs) \cite{Semtech}. At the physical layer, LoRa employs chirp spread spectrum (CSS) modulation, which offers high sensitivity and robustness to interference and multipath, while enabling simple transceivers. These benefits, however, come at the cost of low data rates and modest spectral efficiency. The data rate of LoRa can be adjusted by changing the spreading factor (SF): a smaller SF leads to a shorter symbol duration and thus a higher data rate. However, adjacent SFs differ by a factor of two in symbol duration while increasing the payload by only one bit per symbol, which results in coarse data-rate granularity and modest spectral efficiency. This motivates physical-layer designs that fill the large data-rate gaps between adjacent SFs, increasing throughput and spectral efficiency while preserving LoRa’s key advantages.

Existing attempts to boost LoRa's throughput typically enrich the modulation carried by a single LoRa symbol, for example by exploiting additional dimensions such as phase, amplitude, chirp shape, or index. Another line of work increases throughput by allowing multiple transmitters to share the same time-frequency resources via superposition of their signals. In contrast, we apply superposition at the waveform level. Rather than modifying the symbol alphabet or requiring a new demodulator, we linearly superpose a carefully designed auxiliary waveform onto each legacy LoRa symbol. In this way, we increase data rate and spectral efficiency in a lightweight manner, while the core LoRa demodulation chain can remain unchanged and the overall receiver complexity is not significantly increased.

\subsection{Related Work}

A number of works have analytically characterized the CSS modulation from different aspects. Prior work has characterized CSS/LoRa modulation from multiple angles, including symbol (quasi-)orthogonality and spectral properties \cite{Chiani2019}, inter-SF interference due to imperfect orthogonality \cite{croce2018}, and analytical SER under noise and interference \cite{Orion2019}.

Many works aim to increase LoRa throughput by embedding extra bits into additional modulation dimensions of LoRa symbols. 

PSK-LoRa \cite{PSK2019} encodes information in two dimensions: the conventional LoRa symbol index and an additional \(M\)-ary PSK phase imposed on each symbol. Subsequent extensions further expand the modulation space by exploiting the chirp slope as a third dimension \cite{SSKPSK2023} and by refining the phase mapping and detector design to improve robustness against practical impairments \cite{ePSK_LoRa2021}. Several works exploit the amplitude modulations (AM): QR-LoRa in \cite{qrlora2025} lets the symbol envelope carry extra information via quantized amplitude ramps and CloakLoRa in \cite{cloaklora} uses small amplitude variations on top of a legacy LoRa transmission to build a covert AM channel. There are also index-modulation based designs. Frequency-shift CSS with index modulation (FSCSS-IM) in \cite{FSCSS} encodes bits in the index of orthogonal chirp combinations to increase the data rate of CSS systems. Multiple-chirp-rate index modulation (MCR-IM) for LoRa in \cite{MRCIM2025} and spreading-factor-index-aided LoRa (SFI-LoRa) in \cite{SFI_2025} follow a similar idea: they build a LoRa-compatible codebook of CSS waveforms (e.g., from Chu–Zadoff sequences with different chirp rates or spreading-factor patterns) and map information to the waveform index. Some works depart from linear chirps. CurvingLoRa in \cite{curving2022} replaces linear chirps with non-linear chirps to mitigate energy convergence in collisions and increase concurrent transmission capacity, while BIC-LoRa in \cite{bic2024} further encodes extra bits directly in the chirp shape, using a richer set of non-linear chirps to convey more information.

While these symbol-level modifications of the LoRa waveform can significantly improve spectral efficiency, they typically require non-standard receivers and more complex demodulation algorithms compared with legacy LoRa chips.

Several works apply non-orthogonal multiple access (NOMA) to LoRa or LoRa-like CSS by superposing multiple users in the power domain on the same time-frequency resources. For uplink scenarios, \cite{Uplink2019} considers a NOMA-enabled LPWAN/LoRaWAN and jointly optimizes the radio-resource configuration to improve spectral efficiency and user fairness, while \cite{noma} designs a random-access control scheme for a NOMA-enabled LoRa network that coordinates users’ transmission parameters so that power-domain multiplexing at the gateway can successfully resolve collisions. NOMA-based schemes for LoRa mainly superpose conventional LoRa symbols from different users and focus on optimizing network resources so that multiple links can coexist without disturbing each other. 

Super-LoRa in \cite{superlora} explicitly superposes multiple payload symbols within the same transmission to enhance throughput, leveraging the robustness of the LoRa demodulator. Our work is similar in spirit, as it also exploits superposition within a single LoRa transmission, but we redesign the LoRa symbol itself and superpose a new chirp layer on the legacy waveform to increase spectral efficiency while ensuring both layers could be reliably demodulated.

\subsection{Contributions}

The main contributions of this paper are as follows.
i) We show that no non-zero waveform can be completely transparent to the legacy LoRa demodulator, and we derive conditions under which a superposed waveform induces only marginal degradation.
ii) We prove that a high-SF LoRa waveform superposed on a low-SF symbol approximately satisfies these conditions.
iii) Leveraging this property, we treat a high-SF segment as a quasi-narrowband carrier for a BPSK stream, derive closed-form error-rate expressions for both layers, and validate them via Monte Carlo simulations.

\section{System Model}
\label{II}

\subsection{Properties of LoRa Symbols} 

LoRa uses CSS modulation with configurable bandwidths \(B\) of 125, 250, or 500 kHz. For a spreading factor (SF), each LoRa symbol encodes \(\mathrm{SF}\) bits, yielding \(N = 2^{\mathrm{SF}}\) distinct symbols. The symbol indices range from \(0\) to \(N-1\), and the symbol at index \(0\) is called the \emph{upchirp}. A baseband LoRa symbol \(s \in \{0,1,\dots,N-1\}\) starts at a frequency of \(\frac{sB}{N}-\frac{B}{2}\) and increases over time with a slope of \(\frac{B}{T_s}\). When the instantaneous frequency reaches \(\frac{B}{2}\), it is wrapped to \(-\frac{B}{2}\) and continues increasing at the same slope. 

To analyze the modulation more conveniently, we work with a discrete-time baseband model. The duration of each LoRa symbol is \(T=\frac{2^{\text{SF}}}{B}\). Choosing a sampling rate \(f_s = B\), each symbol with index \(s\) is represented by \(N\) complex samples \(x_s[n]\). The samples can be written as
\begin{equation}
    x_s[n] = e^{j\phi_s[n]}, \quad n=0,\dots,N-1,
\end{equation}
where the symbol phase is given by
\begin{equation}
    \phi_s[n]
    = \frac{2\pi}{N}\Bigl(\frac{n^2}{2} + \bigl(s - \tfrac{N}{2}\bigr)n\Bigr).
    \label{eq:lora_phase_quadratic}
\end{equation}
The folding of the instantaneous frequency at \(\pm \tfrac{B}{2}\) is implicitly handled by the complex exponential.

\begin{lemma}[Completeness and Orthogonality of LoRa Symbols]
For a given spreading factor \(\mathrm{SF}\), the \(N = 2^{\mathrm{SF}}\) LoRa baseband symbols \(\{x_s\}_{s=0}^{N-1}\) form a complete orthogonal basis of \(\mathbb{C}^N\).
\end{lemma}

\begin{proof}
Using \eqref{eq:lora_phase_quadratic}, we factor \(x_s[n]\) as
\begin{equation}
    x_s[n]
    = \exp\!\left\{ j\frac{2\pi}{N}\Bigl(\frac{n^2}{2} - \frac{N}{2}n\Bigr) \right\}
      e^{j 2\pi s \frac{n}{N}}
    = g[n]\, e^{j 2\pi s \frac{n}{N}},
\end{equation}
where \(g[n]\) does not depend on \(s\) and satisfies \(|g[n]| = 1\) for all \(n\).
For \(s_1,s_2 \in \{0,\dots,N-1\}\), their inner product is
\begin{align}
    \langle x_{s_1}, x_{s_2} \rangle
    &= \sum_{n=0}^{N-1} x_{s_1}[n]\, x_{s_2}^*[n]
    = \sum_{n=0}^{N-1} |g[n]|^2 e^{j 2\pi (s_1-s_2)\frac{n}{N}}\\
    &= \sum_{n=0}^{N-1} e^{j 2\pi (s_1-s_2)\frac{n}{N}}= 
    \begin{cases}
        N, & s_1 = s_2,\\[0.3em]
        0, & s_1 \neq s_2,
    \end{cases}
\end{align}
where the last equality follows from the finite geometric series.
Thus \(\{x_s\}_{s=0}^{N-1}\) is a set of \(N\) mutually orthogonal nonzero vectors in
the \(N\)-dimensional space \(\mathbb{C}^N\), and hence forms a complete orthogonal basis.
\end{proof}

Once we know that a waveform cannot vanish in the subspace spanned by the LoRa symbols, the next question is how its energy manifests at the demodulator.

Let \(r[n]\) denote the received signal. The first step of demodulation is \textit{dechirping}, which consists of multiplying \(r[n]\) by the complex conjugate of the \emph{upchirp} \(x_0^*[n]\) for each sampling point \(n\)
\begin{equation}
    y[n] = x_0^*[n]\, r[n].
\end{equation}
An \(N\)-point Discrete Fourier Transform (DFT) is then applied to \(y[n]\) in order to obtain the decision metric \(Y[k]\), which is given by
\begin{equation}
Y[k] = \frac{1}{\sqrt{N}}\sum_{n=0}^{N-1} y[n]\, e^{-j\frac{2\pi}{N}kn},\quad k=0,1,\dots,N-1.
\end{equation}
The transmitted symbol index \(s\) is recovered by finding the DFT bin in \(Y[k]\) with the largest magnitude
\begin{equation}
\hat{s}
\;=\;
\arg\max_{0 \,\le\, k \,<\, N}\;\bigl|Y[k]\bigr|.
\label{decision}
\end{equation}
Ideally, the \(s\)-th LoRa symbol only produces a peak in the \(s\)-th frequency bin in the decision metric \(Y[k]\).

\begin{lemma}[Energy Conservation of Demodulation Process]
Let \(u[n]\) denote an arbitrary waveform, and let \(E_{\mathrm{u}} = \sum_{n=0}^{N-1} |u[n]|^2\) denote its energy. Let \(U[k]\) be the \(N\)-point DFT of the dechirped waveform, then the dechirp-and-DFT operation preserves the energy
\begin{equation}
    E_{\mathrm{u}} = \sum_{n=0}^{N-1} |u[n]|^2
    = \sum_{k=0}^{N-1} |U[k]|^2.
\end{equation}
\end{lemma}

\begin{proof}
Because \(|x_0[n]| = 1\) for all \(n\), multiplication by the upchirp (or its conjugate) is a unit-modulus operation and therefore preserves the energy of \(u[n]\) in the time domain. The \(N\)-point DFT with \(\frac{1}{\sqrt{N}}\)
normalization is a unitary transform, so it also preserves energy. Hence, the energy of the waveform must reappear as energy in the \(N\) DFT bins.
\end{proof}

The next question is how to distribute the energy \(E_u\) of an additional waveform superposed on the LoRa symbol so as to minimize the worst-case per-bin energy. This leads to the following result.

\begin{lemma}[Uniform Allocation Minimizes Worst-Case Per-Bin Energy]
\label{lemma3}
Consider complex sequence \(\{U[k]\}_{k=0}^{N-1}\) with total energy constraint
\begin{equation}
    \sum_{k=0}^{N-1} |U[k]|^2 = E_{\mathrm{u}}.
\end{equation}
Then the minimum achievable worst-case per-bin interference power
\(\max_{0 \le k < N} |U[k]|^2\) is
\begin{equation}
    \min_{\{U[k]\}} \max_{0 \le k < N} |U[k]|^2
    = \frac{E_{\mathrm{u}}}{N},
\end{equation}
and this minimum is attained if and only if
\begin{equation}
    |U[0]|^2 = |U[1]|^2 = \dots = |U[N-1]|^2 = \frac{E_{\mathrm{u}}}{N}.
\end{equation}
\end{lemma}

\begin{proof}
By the inequality between the maximum and the arithmetic mean, for any
nonnegative sequence \(\{a_k\}_{k=0}^{N-1}\),
\begin{equation}
    \max_k a_k
    \;\ge\; \frac{1}{N} \sum_{k=0}^{N-1} a_k.
\end{equation}
Applying this to \(a_k = |U[k]|^2\) and using
\(\sum_{k=0}^{N-1} |U[k]|^2 = E_{\mathrm{u}}\), we obtain
\begin{equation}
    \max_k |U[k]|^2
    \;\ge\; \frac{1}{N} \sum_{k=0}^{N-1} |U[k]|^2
    = \frac{E_{\mathrm{u}}}{N}.
\end{equation}
Equality holds if and only if \(|U[k]|^2\) is constant over \(k\). Thus the uniform allocation uniquely minimizes the worst-case per-bin energy under the total-energy constraint.
\end{proof}

\subsection{High-SF LoRa Waveform on a Low-SF Symbol}

We now specialize the above analysis to a concrete and practically relevant case: superposing a (part of) higher-SF LoRa waveform on top of a lower-SF LoRa symbol, and examine how the resulting pattern of decision metric approximates the uniformly distributed case discussed above. 

Suppose the spreading factors are \(\mathrm{SF}_l\) and \(\mathrm{SF}_h\) with \(h>l\). Assume that, over one \(\mathrm{SF}_l\) symbol interval, the contribution of this \(\mathrm{SF}_h\) symbol to the received signal is given by a contiguous segment of length \(N_l\) starting at some offset \(n_s \in \{0,\dots,N_h-N_l\}\). After dechirping with the \(\mathrm{SF}_l\) upchirp and applying an \(N_l\)-point DFT, the contribution by an \(\mathrm{SF}_h\) waveform to bin \(k\) can be written as
\begin{equation}
    U[k]
    = \frac{1}{\sqrt{N_l}}\sum_{n=0}^{N_l-1}
        \exp\!\left\{
            j\pi p n^2 + j 2\pi q_k n +C
        \right\},
    \label{Uk}
\end{equation}
where \(k=0,\dots,N_l-1\), \(N_l = 2^{\mathrm{SF}_l}\) and \(N_h = 2^{\mathrm{SF}_h}\) denote the numbers of samples per symbol for the low-SF and high-SF chirps, respectively, and \(C\) is a term which is not related to \(n\). The coefficients are defined as
\begin{equation}
    p = \frac{N_l - N_h}{N_l N_h} \neq 0,
    \qquad
    q_k = \frac{n_s+s_h}{N_h} - \frac{k}{N_l},
    \label{eq:a_bk_def}
\end{equation}
where \(s_h\) is the \(\mathrm{SF}_h\) symbol index.

\begin{lemma}[High-SF Waveform Seen by a Low-SF Demodulator]
\label{lem:highSF_wideband}
Let \(U[k]\) be defined as in \eqref{Uk}. Then there exists a contiguous index set
\(K \subset \{0,\dots,N_l-1\}\) with cardinality
\begin{equation}
    |K| \approx N_l\Bigl(1 - \frac{N_l}{N_h}\Bigr),
\end{equation}
such that
\begin{equation}
    |U[k]| \approx \frac{1}{\sqrt{N_l \lvert p \rvert}}, \qquad k \in K,
\end{equation}
while \(|U[k]|\) is negligible for \(k \notin K\).
\end{lemma}

\begin{proof}
See Appendix.
\end{proof}

For large-SF gaps with \(N_h \gg N_l\), the ratio \(\frac{N_l}{N_h}\) becomes small and \(|K|\approx N_l\), so the higher-SF waveform appears to the low-SF demodulator as an almost flat wideband interferer over all \(N_l\) decision bins, closely matching the near-uniform energy allocation identified in Lemma~\ref{lemma3}. For the example with SF7 and SF12 in Figure~\ref{fig:underlay}, almost all bins (except the bin corresponding to the symbol index) in the magnitude of the decision metric \(|Y_k|\) contain approximately the same nonzero value.

\begin{figure}[!t]
  \centering
  \subfloat[Spectrogram of a single SF7 LoRa symbol.]{
    \includegraphics[width=0.46\linewidth]{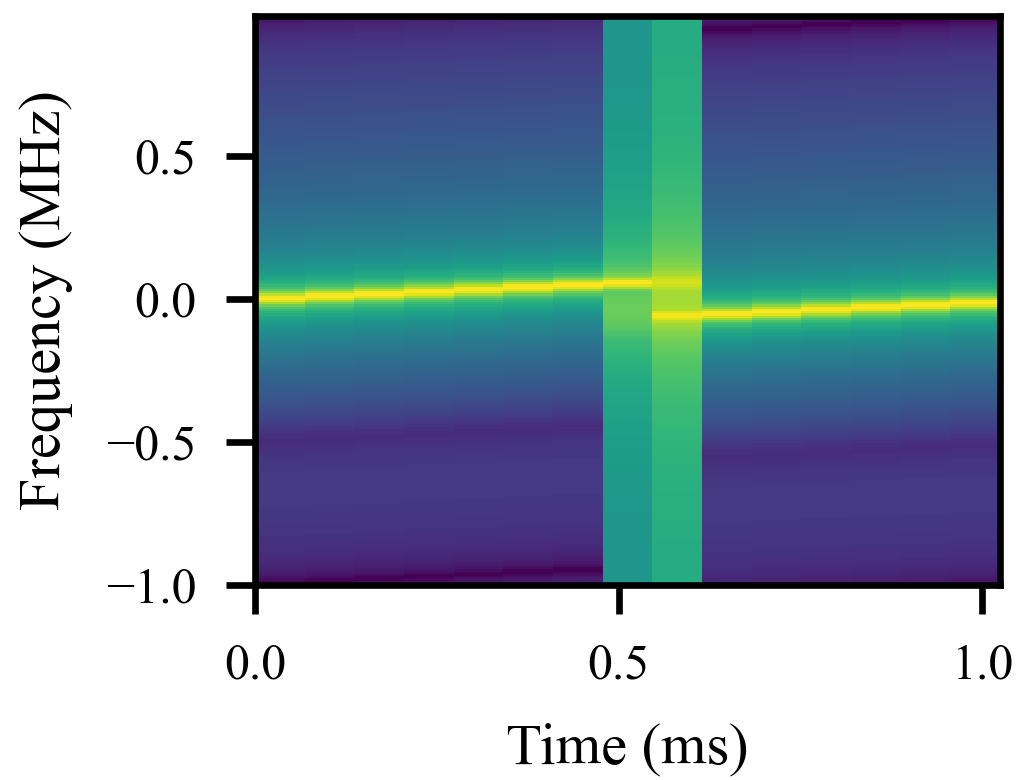}%
    \label{fig:underlay:a}
  }\hfill
  \subfloat[Spectrogram of the SF12 segment embedded within one SF7 symbol interval.]{
    \includegraphics[width=0.46\linewidth]{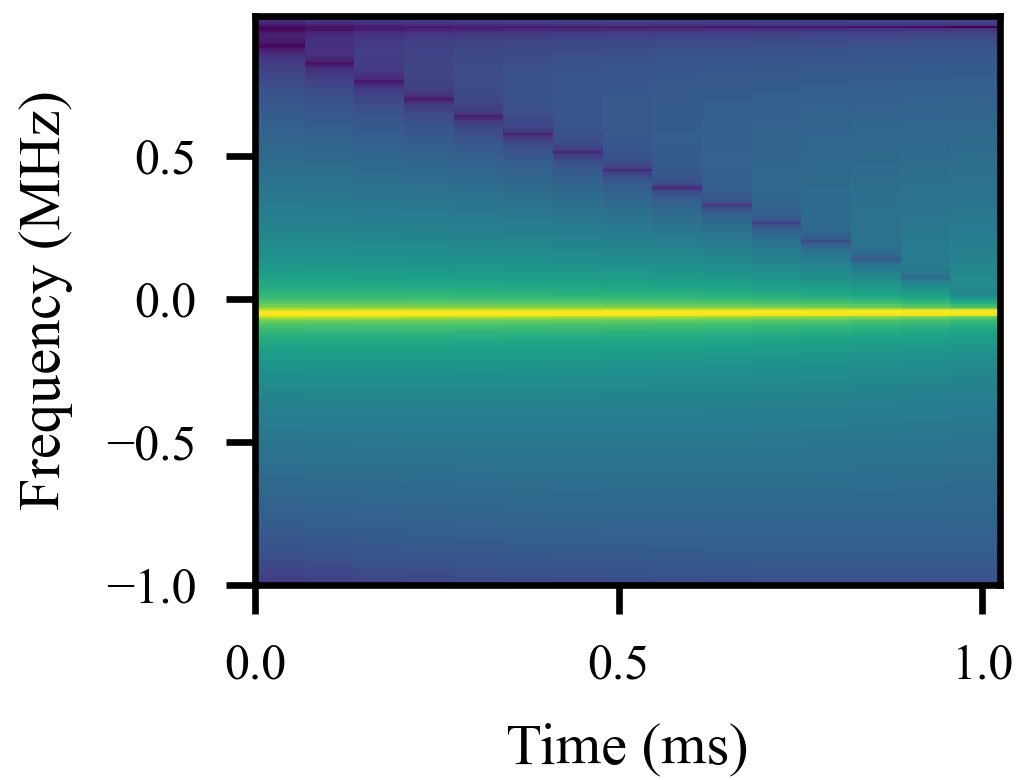}%
    \label{fig:underlay:b}
  }\\[2pt]
  \subfloat[Superposition of the SF7 symbol and the SF12 segment.]{
    \includegraphics[width=0.46\linewidth]{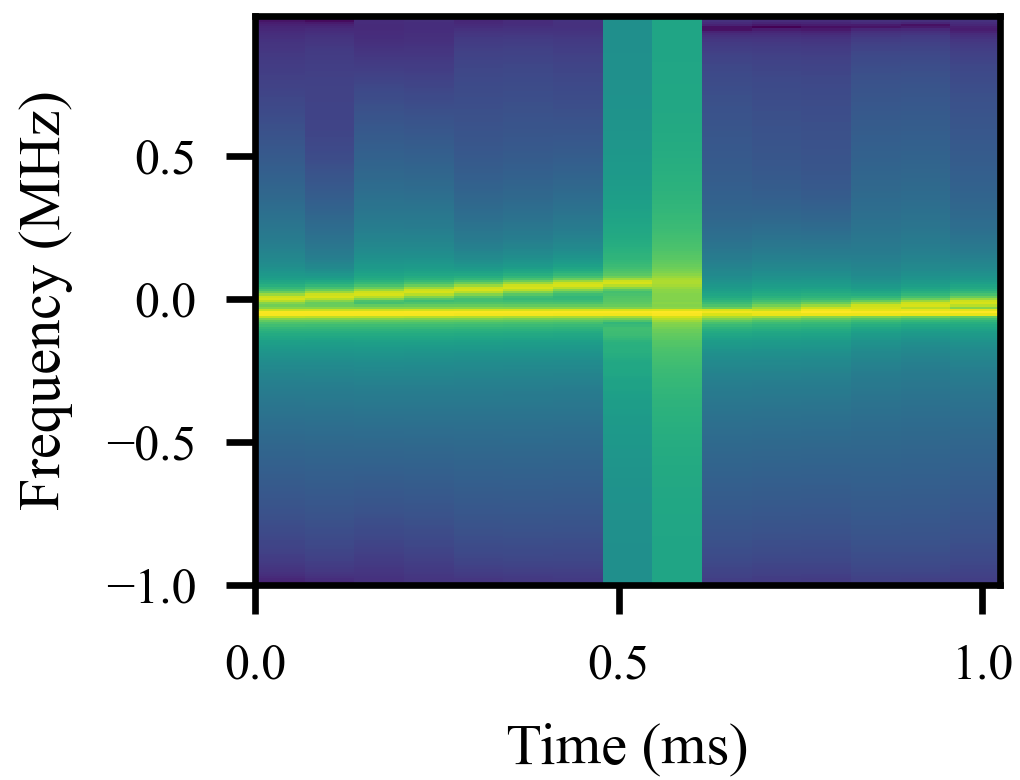}%
    \label{fig:underlay:c}
  }\hfill
  \subfloat[Magnitude of decision metric after SF7 dechirp-and-DFT demodulation.]{
    \includegraphics[width=0.46\linewidth]{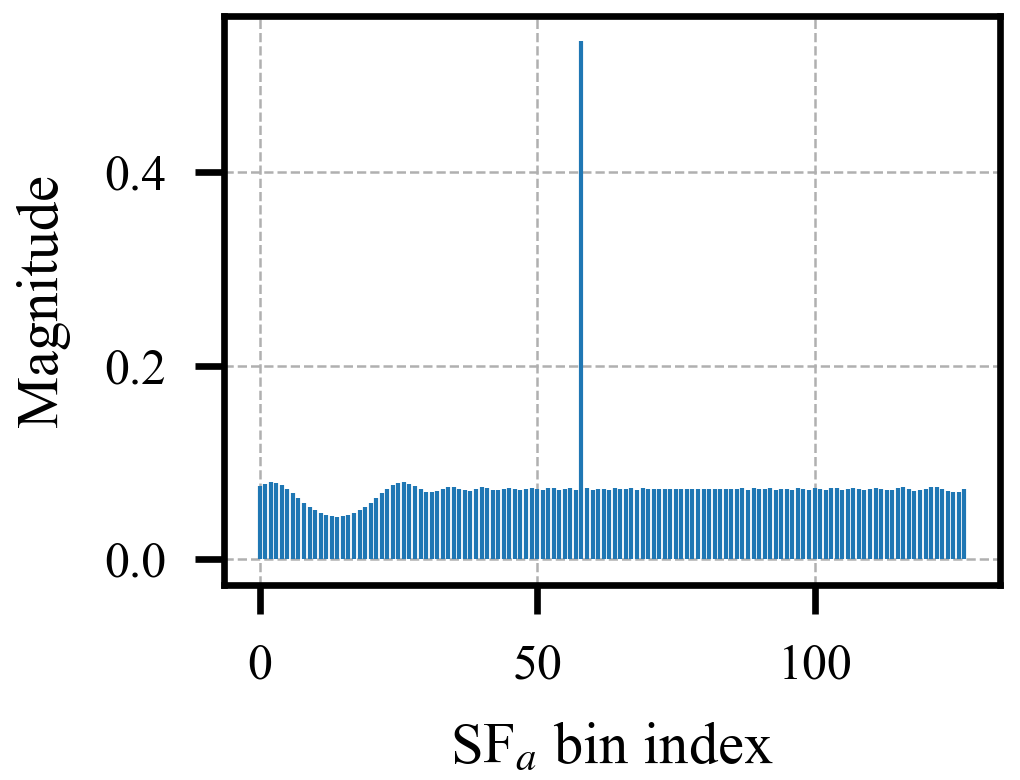}%
    \label{fig:underlay:d}
  }
  \caption{Illustration of a SF12 LoRa waveform embedded within a SF7 symbol. Subfigures~(a)--(c) show the SF7 symbol, the SF12 segment within one SF7 symbol interval, and their superposition, respectively. Subfigure~(d) shows the magnitude of decision metric \(\lvert Y[k]\rvert\) after SF7 dechirp-and-DFT demodulation. }
  \label{fig:underlay}
\end{figure}

\subsection{Proposed High-SF Superposition Scheme}

\begin{figure*}[!t]
    \centering
    \includegraphics[width=0.9\linewidth]{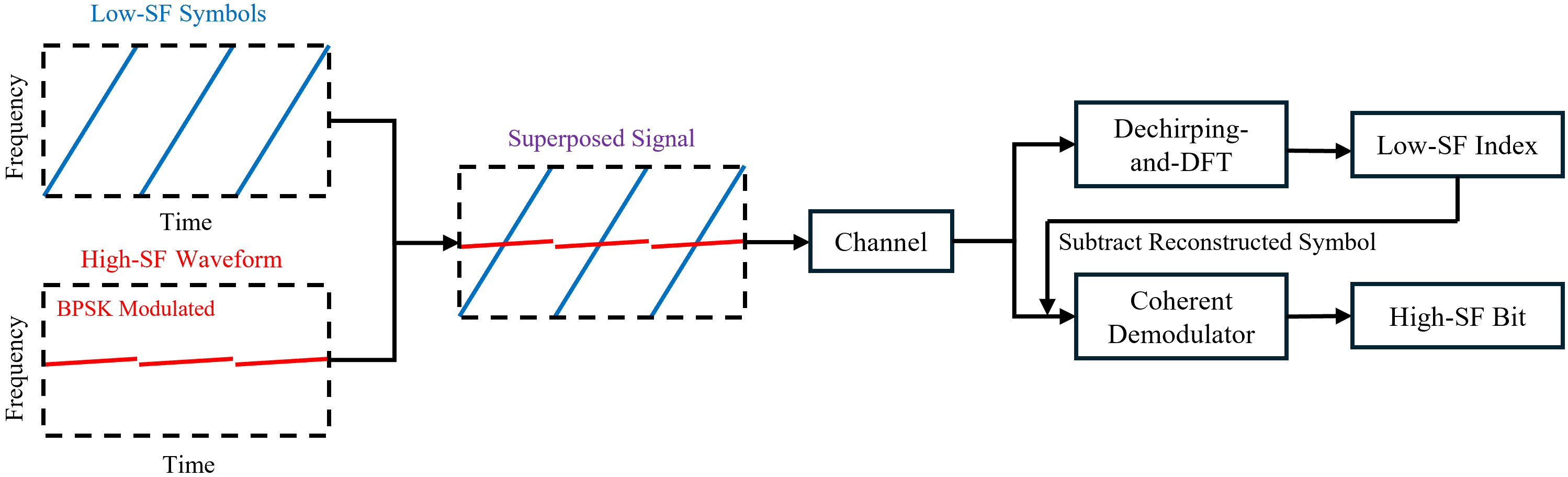}
    \caption{Modulation and demodulation of proposed scheme.}
    \label{fig:long_figure}
\end{figure*}

Building on the above characterization, we now describe our design as shown in Figure~\ref{fig:long_figure}. The key idea is to reuse a chunk of high-SF LoRa \emph{upchirp} (LoRa symbol with index 0) as the carrier of an additional BPSK stream, and to superpose this stream on top of a low-SF LoRa symbol. Within a single low-SF symbol, the instantaneous frequency of the high-SF upchirp changes only slightly. From the perspective of the low-SF dechirp-and-DFT demodulator, the superposed energy therefore spreads approximately uniformly across decision bins. This approximation is most accurate for large SF gaps (\(N_h \gg  N_l\)), where the high-SF segment spreads energy nearly uniformly across the low-SF DFT bins, as quantified in the previous subsection.

Motivated by this near-uniform spreading behavior, we split the high-SF upchirp into segments aligned with the low-SF symbol duration and embed information by applying BPSK phase rotations to each segment. Choosing a given high-SF segment is equivalent to selecting a corresponding frequency channel within the low-SF band. Let \(x_0'[n]\) denote the high-SF upchirp segment with normalized power and let \(c\) be the complex data symbols drawn from a BPSK constellation. The high-SF signal (with normalized power) over the one low-SF symbol is then given by
\begin{equation}
    x_{h}[n] = c\, x_0'[n], \quad n = 0,\dots,N_l-1.
\label{eq:uul}
\end{equation}
The total transmitted signal is the sum of the low-SF LoRa symbol and the appropriately power-scaled waveform \(x_h[n]\).

At the receiver, we first demodulate the low-SF LoRa symbol by a regular dechirp-and-DFT demodulator. Using the detected low-SF symbols, the receiver reconstructs the corresponding baseband low-SF waveform and subtracts this reconstructed waveform from the received mixture. After this cancellation step, the residual signal ideally contains only the high-SF segment plus noise. The residual high-SF component behaves, over each low-SF symbol interval, approximately as a narrowband tone whose center frequency can be fixed and agreed upon a priori by choosing a fixed high-SF segment. The demodulation path of the superposed layer therefore operates on this residual signal: it first passes the residual through a band-pass filter, and then performs coherent BPSK demodulation by correlating the filtered signal with the known high-SF segment. Since the LoRa preamble and header already provide timing and frequency synchronization, the receiver for superposed waveform can reuse this synchronization information. In this architecture, the high-SF waveform both mildly interferes with the low-SF symbol and serves as a coherent carrier for the BPSK stream.

\section{Error Probability Analysis with High-SF Waveform}
\label{sec:error_analysis}

For practical hardware implementations, the received signal is typically oversampled to facilitate sampling time offset (STO) compensation and symbol-boundary detection. Let the oversampling factor be denoted by \(\beta\).

\subsection{Symbol Error Rate of LoRa Demodulation}

Let \(P_{l}\) denote the power of low-SF LoRa signal and \(P_{\text{n}}\) denote the noise power, the \emph{baseline} signal-to-noise ratio (SNR) is defined as
\begin{equation}
    \gamma =  \frac{P_l}{P_{\text{n}}}.
\label{SNR}
\end{equation}

The LoRa symbols are identified and segmented; after down-sampling, each symbol is comprised of \(N_l\) samples. At the final step of LoRa demodulation, the receiver applies dechirping followed by an \(N_l\)-point DFT, and decides on the symbol index by selecting the DFT bin with the largest magnitude according to \eqref{decision}.  For a transmitted symbol \(s_l \in
\{0,\dots,N_l-1\}\), the dechirp-and-DFT operation maps all LoRa energy into the \(s_l\)-th bin, while the additive noise is spread across all \(N_l\) bins.
Let \(Y_k\) denote the complex output of bin \(k\). We can write
\begin{equation}
  Y_k =
  \begin{cases}
    \sqrt{N_lP_{l}} + W_k, & k = s_l,\\[0.3em]
    W_k,       & k \neq s_l,
  \end{cases}
\end{equation}
where \(k = 0,\dots,N_l-1\) and \(W_k \sim \mathcal{CN}(0,P_{\text{n}})\) are complex Gaussian noise with power \(P_{\text{n}}\), and the factor $\sqrt{N_l P_{l}}$ follows from conservation of the time-domain signal energy under an \(N_l\)-point unitary DFT normalization.

For noise-only bins, the magnitude \(R_k=|Y_k|\) follows a Rayleigh distribution. For an amplitude value \(r \ge 0\), the probability density function (PDF) and cumulative distribution function (CDF) are
\begin{align}
  f_{\text{Ra}}(r;\gamma)
    &= \frac{2r\gamma}{P_{l}}
       \exp\!\left(-\frac{r^2\gamma}{P_{l}}\right), \\
  F_{\text{Ra}}(r;\gamma)
    &= 1 - \exp\!\left(-\frac{r^2\gamma}{P_{l}}\right).
\end{align}

The signal bin \(k = s_l\) contains both signal and noise. Since \(Y_{s_l} \sim \mathcal{CN}\!\bigl(\sqrt{N_l P_{l}},\, \frac{P_l}{\gamma}\bigr)\), its amplitude \(R_{s_l}=|Y_{s_l}|\) follows a Rice distribution. For an amplitude value \(r \ge 0\), the probability density function is
\begin{equation}
  f_{\text{Ri}}(r;\gamma)
    = \frac{2r\gamma}{P_{l}}
      \exp\!\left(-\frac{r^2\gamma + N_l P_{l}\gamma}{P_{l}}\right)
      I_0\!\left(\frac{2r\gamma\sqrt{N_l P_{l}}}{P_{l}}\right),
\end{equation}
where \(I_0(\cdot)\) is the modified Bessel function of the first kind.

Let
\begin{equation}
    R'_{\max} = \max_{k \neq s_l} R_k
\end{equation}
denote the maximum over the \((N_l-1)\) noise-only bins. Since the \(R_k\) for \(k \neq s_l\) follows Rayleigh distribution, the CDF of \(R'_{\max}\) is
\begin{equation}
  F_{R'_{\max}}(r;\gamma)
    = \Pr\{R'_{\max} \le r\}
    = \bigl[F_{\mathrm{Ra}}(r;\gamma)\bigr]^{N_l-1}.
\end{equation}

A correct decision occurs when the signal bin magnitude exceeds the largest noise bin magnitude, yielding the correct decision probability
\begin{align}
  P_c(\gamma)
    &= \Pr\{R_{s_l} > R'_{\max}\}\\
    &= \int_0^\infty
        \Pr\{R'_{\max} < r\}\,
        f_{\mathrm{Ri}}\bigl(r;\gamma)\,
      \mathrm{d}r\\
    &= \int_0^\infty
        \bigl[F_{\mathrm{Ra}}(r;\gamma)\bigr]^{N_l-1}
        f_{\mathrm{Ri}}\bigl(r;\gamma)\,
      \mathrm{d}r.
\end{align}
Therefore, the symbol error rate (SER) of ideal LoRa in an AWGN channel can be expressed as
\begin{equation}
  P_e
    = 1 - P_c(\gamma)
    = 1 - \int_0^\infty
          \bigl[F_{\mathrm{Ra}}(r;\gamma)\bigr]^{N_l-1}
          f_{\mathrm{Ri}}\bigl(r;\gamma)\,
        \mathrm{d}r.
  \label{eq:SER}
\end{equation}

Assume the power of high-SF segment as \(P_{h}\), the ratio between the power of low-SF and high-SF waveform (LHR) is given by
\begin{equation}
    \kappa = \frac{P_l}{{P_h}}.
\end{equation}
As discussed in Section~\ref{II}, when a part of high-SF symbol is superposed on top of a low-SF link, its energy is spread approximately uniformly across the decision bins of the low-SF demodulator. It is therefore natural to model the high-SF waveform as an additional white Gaussian noise with power \(P_h\). Under this model, we define the effective SNR for low-SF signal in the presence of the high-SF waveform as

\begin{equation}
    \gamma_l
  = \frac{P_{l}}{P_{h} + P_{\text{n}}}
  = \frac{\gamma\,\kappa}{\gamma+\kappa}.
\label{snr_eff}
\end{equation}

\subsection{Bit Error Rate of Superposition Layer}
To demodulate the superposed waveform, we first reconstruct the low-SF LoRa symbol and subtract it from the received signal, and then filter out-of-band noise by a band-pass filter. The superposed-layer analysis assumes ideal reconstruction and subtraction of the low-SF waveform; residual cancellation error would introduce additional interference at the correlator and increase the required SNR.

Under the assumption of perfect cancellation, the resulting signal can be written as
\begin{equation}
  r'[n] = \sqrt{P_{h}}\,c\,x_0'[n] + w[n],\quad n = 0,\dots,\beta N_l-1,
\end{equation}
with normalized high-SF waveform \(x_0'[n]\) and AWGN \(w[n]\). 

To demodulate the superposed waveform, we apply a coherent correlator with template \(x_0'[n]\) to obtain the complex decision variable
\begin{align}
  z &= \sum_{n=0}^{\beta N_l-1} x_0'^{\ast}[n]\,r'[n]\\
    &= \sqrt{P_{h}}\,c \sum_{n=0}^{\beta N_l-1} |x_0'[n]|^2
      + \sum_{n=0}^{\beta N_l-1} x_0'^{\ast}[n] w[n].
\end{align}
With the normalized waveform \(x_0'[n]\), we have \(\sum_{n} |x_0'[n]|^2 \approx \beta N_l\), where \(\beta\) is the oversampling factor. The useful signal term then has squared magnitude on the order of \(P_{h} \beta^2N_l^2\), whereas the accumulated noise power scales as \(P_\text{n}\beta N_l\). Hence the effective SNR for the high-SF superposed waveform becomes
\begin{equation}
  \gamma_{h}
  = \frac{P_{h} \beta N_l^2}{P_\text{n} \beta N_l}
  = \frac{P_{h}}{P_\text{n}}\,\beta N_l
  = \frac{\gamma}{\kappa}\,\beta N_l.
  \label{eq:snr_ul_beta}
\end{equation}

Therefore, the bit error rate (BER) is given by
\begin{equation}
  P_b
  = Q\!\bigl(\sqrt{2\gamma_{h}}\bigr)
  = Q\!\biggl(\sqrt{2\frac{\gamma}{\kappa}\,\beta N_l}\biggr),
  \label{eq:BER}
\end{equation}
where \(Q(\cdot)\) is the standard Gaussian \(Q\)-function.

\subsection{Feasible Region for \(\gamma\) and \(\kappa\)}

To enable simultaneous data transmission for both waveforms, we target a region of operation where both signals achieve reliable demodulation performance. We take \(\text{SF}_l=7\) and \(\text{SF}_h=12\) as examples. Let the oversampling factor be \(\beta = 16\); that is, for a LoRa bandwidth of \(B = 125\,\mathrm{kHz}\), the device sampling rate is \(2\,\mathrm{MHz}\).

According to \cite{Semtech}, the minimum SNR required for successful LoRa SF7 demodulation is approximately \(-6\,\mathrm{dB}\), which means

\begin{equation}
    \gamma_{l}
    = \frac{\gamma\,\kappa}{\gamma+\kappa}\ge -6\text{ dB}.
\label{eq:lora_constraint}
\end{equation}

For the high-SF superposition transmission, we impose a target BER of \(P_b \le 10^{-5}\), so

\begin{equation}
    P_b = Q\!\biggl(\sqrt{2\frac{\gamma}{\kappa}\,\beta N_l}\biggr)\le 10^{-5}.
\label{eq:bpsk_constraint}
\end{equation}

Constraints \eqref{eq:lora_constraint} and \eqref{eq:bpsk_constraint} jointly define the feasible SNR region for two-waveform coexistence. As Figure~\ref{fig:feasible} shows, the feasible region can be visualized by sweeping \((\gamma,\kappa)\) and plotting the boundary where both inequalities hold. 

\begin{figure}[htbp]
\centerline{\includegraphics[width=1\linewidth]{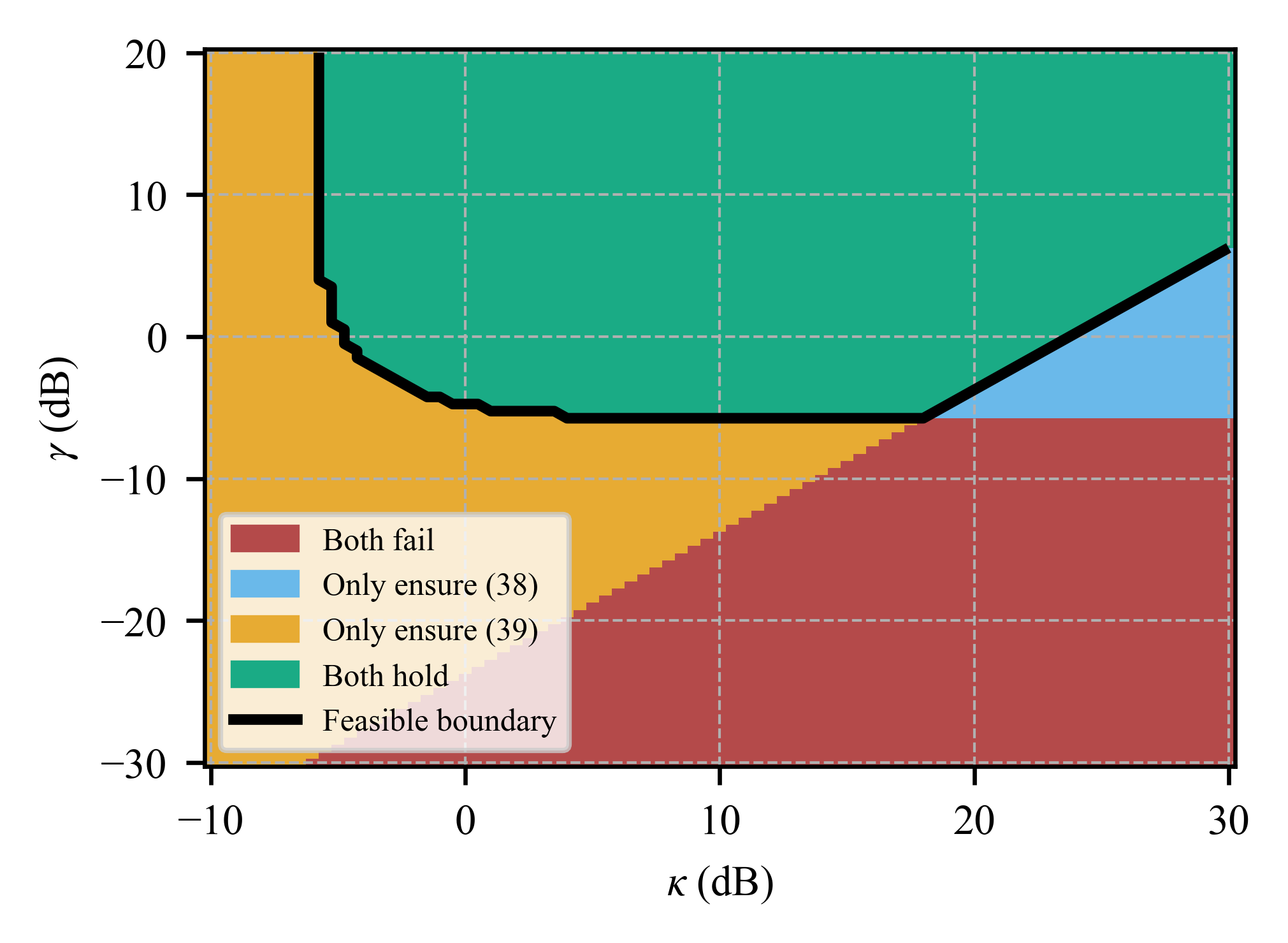}}
\caption{Feasible region of baseline SNR \(\gamma\) and LHR \(\kappa\). Colors indicate which of the low-SF and high-SF constraints are satisfied, and the solid curve marks the boundary where both constraints hold.}
\label{fig:feasible}
\end{figure}

\section{Simulation Results}

We now compare the analytical SER and BER expression in \eqref{eq:SER} and \eqref{eq:BER} with simulation results for a SF7 symbol in the presence of one SF12 superposed waveform. In all experiments, the SF7 and SF12 signals are transmitted over an AWGN channel and the oversampling factor of the receiver is \(16\). All the markers indicate Monte-Carlo simulation results obtained over \(10^5\)~transmitted symbols.

\subsection{Symbol Error Rate of Low-SF Layer}

\subsubsection{\textbf{SER vs. \(\gamma\) for Different \(\kappa\)}}

Figure~\ref{SER_SNR} shows the SER as a function of the baseline SNR~\(\gamma\) for three LHRs~\(\kappa\) with a single superposed SF12 waveform. For each LHR value, the solid curves show the theoretical SER from~\eqref{eq:SER} evaluated with the effective SNR for SF7 signal. Here, \(\kappa = +\infty\) corresponds to the case with no superposed signal.

\begin{figure}[htbp]
\centerline{\includegraphics[width=1\linewidth]{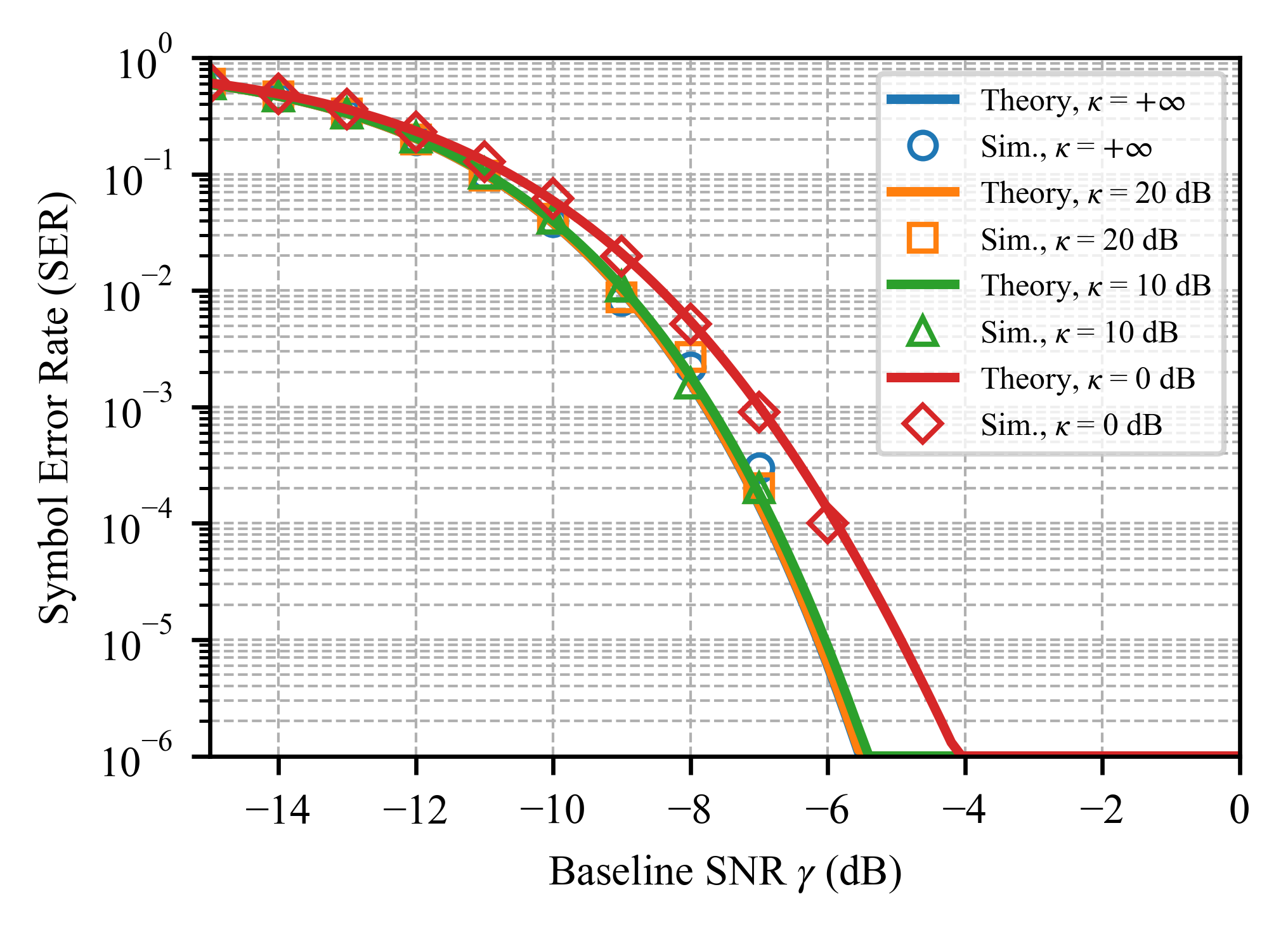}}
\caption{Simulated SER of SF7 symbols with one SF12 superposed waveform as a function of baseline SNR \(\gamma\) for various LHRs \(\kappa\), compared with the theoretical curves.}
\label{SER_SNR}
\end{figure}

Figure~\ref{SER_SNR} shows that the theoretical SER curves closely match the simulation results, indicating that the effective-SNR model in~\eqref{snr_eff} accurately captures the impact of the superposed waveform.

\subsubsection{\textbf{Collapse in the Effective-SNR Domain}}

To further validate the interpretation of the high-SF waveform as additional wideband noise, we re-plot all simulation points in terms of the effective SNR for low-SF defined in~\eqref{snr_eff}. Figure~\ref{SER_effSNR} shows the resulting scatter plot, together with the theoretical SER curve.

\begin{figure}[htbp]
\centerline{\includegraphics[width=1\linewidth]{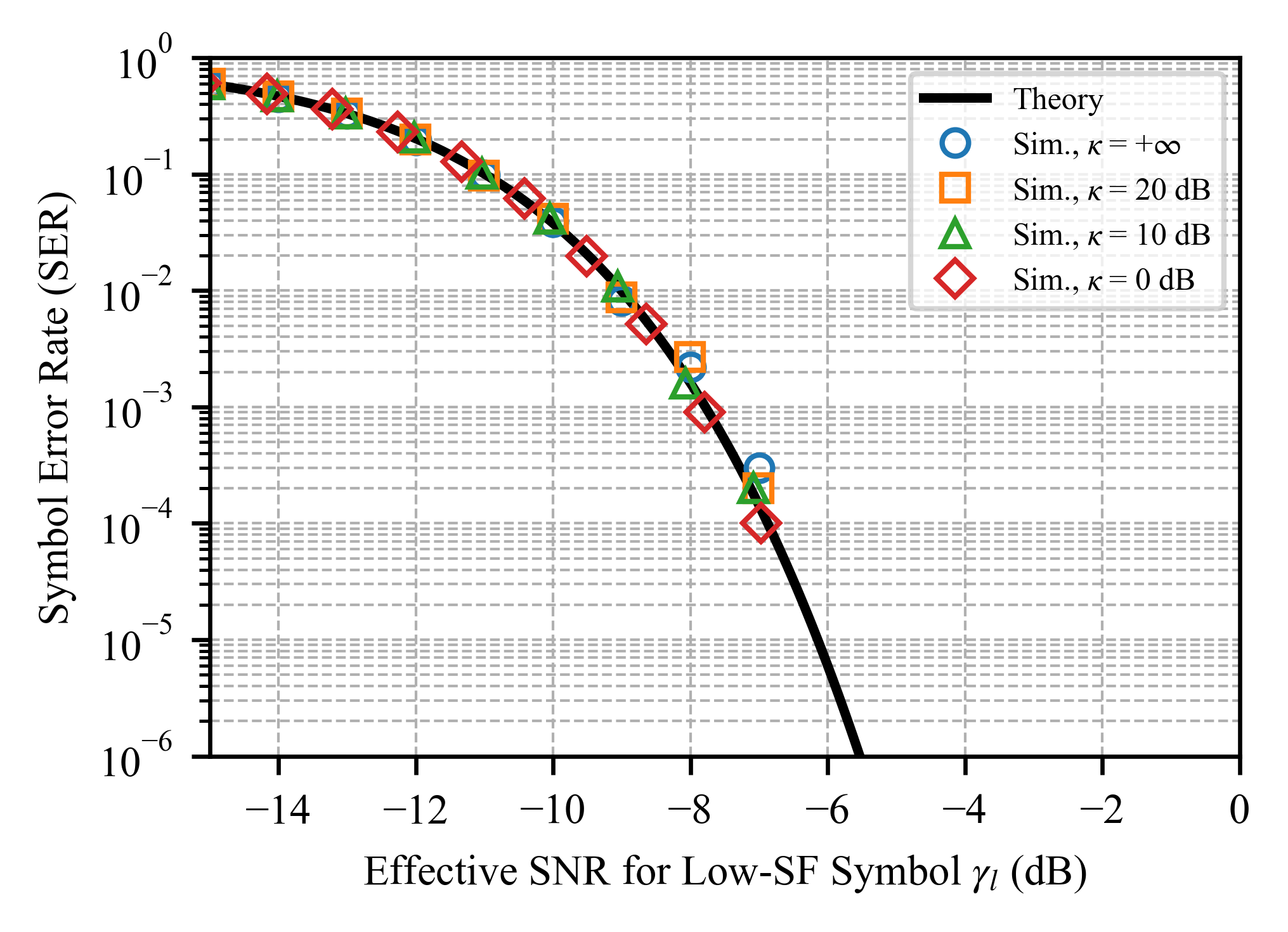}}
\caption{Simulated SER of SF7 symbols with one SF12 superposed waveform as a function of effective SNR \(\gamma_{l}\) for various LHRs \(\kappa\), compared with the theoretical curve.}
\label{SER_effSNR}
\end{figure}

We observe that, once expressed in terms of the effective SNR \(\gamma_l\), all SER measurements taken under different superposed power collapse onto a single curve that closely follows the theoretical curve. This collapse supports the interpretation that, as seen by the low-SF demodulator, the high-SF waveform primarily manifests as an elevated effective noise floor.

\subsection{Bit Error Rate of Superposition Layer}

\subsubsection{\textbf{BER vs. \(\gamma\) for Different \(\kappa\)}}

Figure~\ref{fig:BER_SNR} depicts the theoretical and simulated BER of the high-SF waveform as a function of the baseline SNR \(\gamma\) for various LHRs \(\kappa\). For each curve, \(\kappa\) is kept
constant, so that increasing \(\gamma\) is equivalent to increasing the power of high-SF waveform. The results also show a good agreement between simulations and the theoretical analysis.

\subsubsection{\textbf{Collapse in the Effective-SNR Domain}}

We also re-plot the simulation results as a function of this effective SNR for high-SF waveform according to \eqref{eq:snr_ul_beta} and \eqref{eq:BER}. As shown in Figure~\ref{fig:BER_effSNR}, all BER results obtained under various \(\kappa\) lie very close to the theoretical curve, confirming the correctness of our analysis.

\begin{figure}[htbp]
\centerline{\includegraphics[width=1\linewidth]{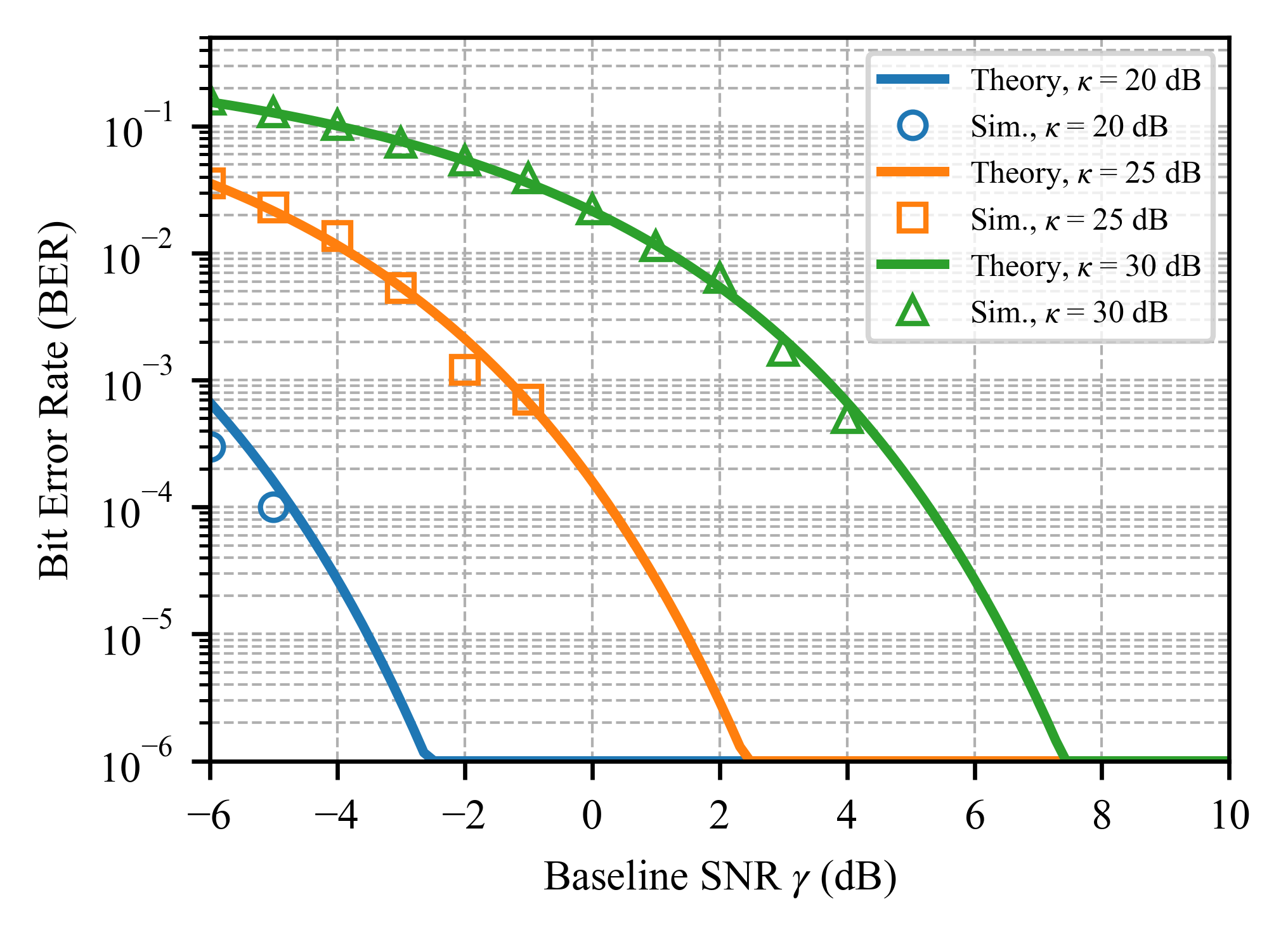}}
\caption{Simulated BER of SF12 waveform as a function of baseline SNR \(\gamma\) for various LHRs \(\kappa\), compared with the theoretical curves.}
\label{fig:BER_SNR}
\end{figure}

\begin{figure}[htbp]
\centerline{\includegraphics[width=1\linewidth]{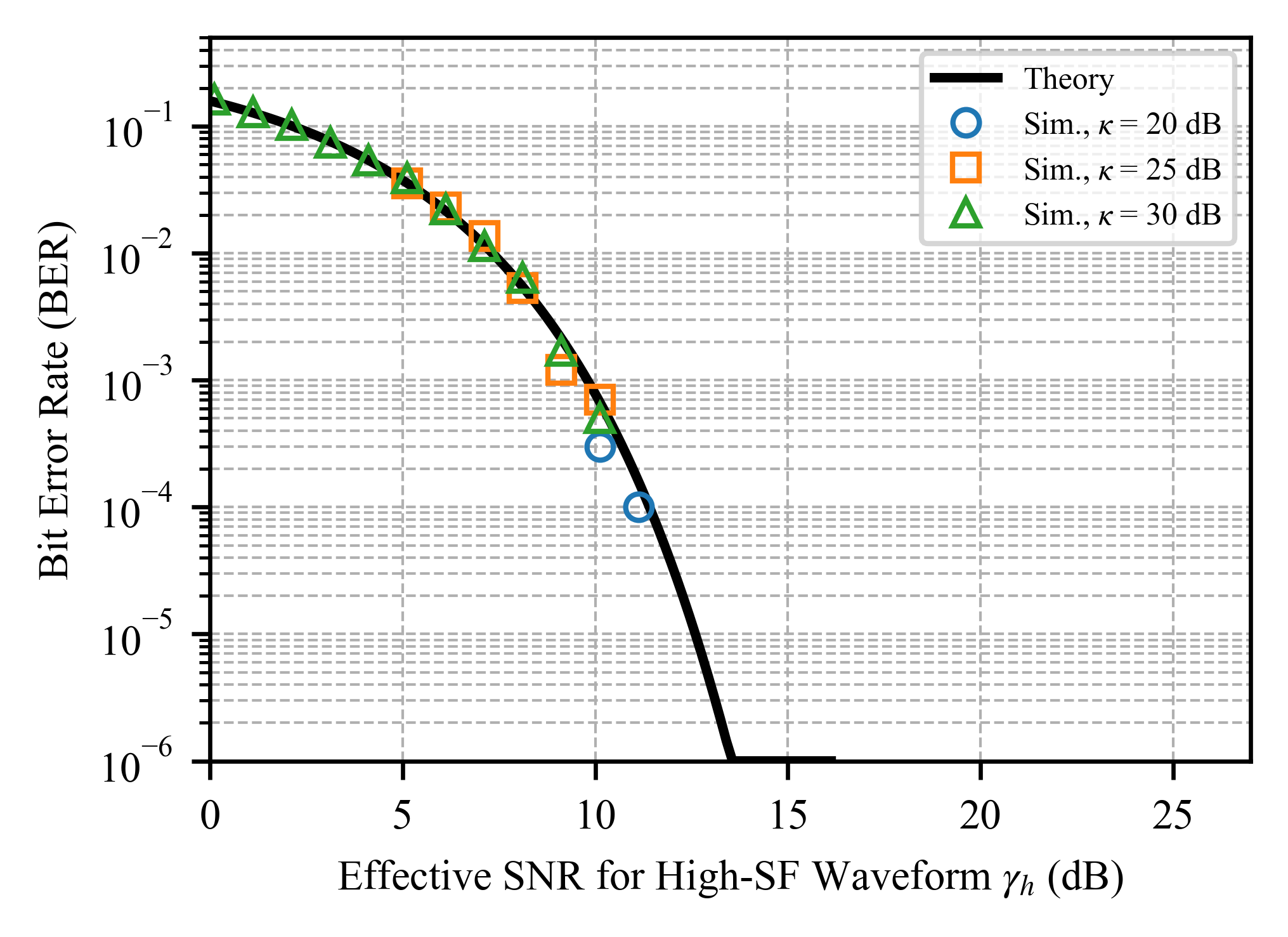}}
\caption{Simulated BER of SF12 waveform as a function of effective SNR \(\gamma_h\) for various LHRs \(\kappa\), compared with the theoretical curve.}
\label{fig:BER_effSNR}
\end{figure}

\section{Conclusion and Future Work}

\subsection{Conclusion}

This paper proposes a waveform-level superposition scheme that overlays an additional signal on top of a standard LoRa symbol, thereby filling the gaps in data rate and spectral efficiency between adjacent spreading factors. We first prove that no non-zero waveform can be completely transparent to the dechirp-and-DFT demodulator, and then derive the conditions under which a superposed signal induces only marginal degradation. Building on this result, we show that a high-SF LoRa waveform is a suitable choice for the superposed layer, as its energy spreads approximately uniformly across the low-SF decision bins, closely matching the theoretically optimal uniform allocation identified by our analysis. This near-uniform spreading allows the high-SF interference to be accurately modeled as an elevated noise floor, leading to a simple effective-SNR characterization for the low-SF layer. Under this model, we derive closed-form error-rate expressions for both the standard LoRa layer and the superposed BPSK stream, characterizing their reliability as functions of the power allocation ratio and channel conditions. Monte-Carlo simulations confirm strong agreement with the analytical predictions. The results demonstrate that, by appropriately adjusting the power ratio between the two layers, reliable transmission can be achieved for both layers simultaneously with only a modest penalty on the legacy LoRa link.

\subsection{Future Work}

We will extend the analysis to multiple simultaneously superposed signals. Because the instantaneous frequency of a high-SF waveform varies slowly within one low-SF symbol interval, different segments of the high-SF upchirp correspond to quasi-narrowband sub-channels at different center frequencies, suggesting a frequency-division style multiplexing of multiple independent superposed streams. The number of approximately non-overlapping segments scales on the order of \(2^{SF_h - SF_l}\), but in practice is limited by the available power budget, guard requirements, and spectral leakage between adjacent sub-channels. Quantifying the inter-stream interference and optimizing the power allocation across these sub-channels will be a key analytical focus. We will also consider higher-order modulations such as QPSK for the superposed layer, which would further improve spectral efficiency at the cost of increased receiver complexity and a higher SNR requirement. The trade-off between modulation order, power allocation ratio, and error-rate performance of both layers needs detailed investigation. In addition, the present analysis assumes an AWGN channel and perfect reconstruction of the low-SF waveform for successive cancellation; extending the error-rate analysis to fading channels and quantifying the sensitivity to imperfect cancellation are important steps toward assessing the scheme's robustness in realistic deployments.  Finally, we plan to validate the scheme on software defined radio (SDR) platforms with over-the-air experiments, evaluating the impact of practical impairments such as carrier frequency offset and power amplifier nonlinearity, and providing a broader comparison against existing LoRa enhancement techniques.

\appendix
\section{Appendix A}
\label{appa}

After the low-SF demodulator, the contribution of high-SF waveform in bin \(k\) is the quadratic exponential sum
\begin{equation}
    U[k]
    = \frac{1}{\sqrt{N_l}}\sum_{n=0}^{N_l-1}
        \exp\!\left\{
            j\pi p n^2 + j 2\pi q_k n + C
        \right\},
\end{equation}
with
\begin{equation}
    p = \frac{N_l - N_h}{N_l N_h}, 
    \qquad
    q_k = \frac{n_s + s_h}{N_h} - \frac{k}{N_l},
\end{equation}
and phase function
\begin{equation}
    \phi_k[n] = \pi p n^2 + 2\pi q_k n + C,
\end{equation}
where \(C\) is not related to \(n\).
We approximate the sum by the integral
\begin{equation}
    U[k]
    \approx \frac{1}{\sqrt{N_l}}\int_{0}^{N_l}
        \exp\!\left\{ j\phi_k(t) \right\}\,\mathrm{d}t,
    \label{eq:U_integral_approx_short}
\end{equation}
which is of the standard stationary-phase form
\(\int g(t)\,e^{j\lambda f_k(t)}\mathrm{d}t\) with
\begin{equation}
    g(t) \equiv \frac{1}{\sqrt{N_l}},\quad
    f_k(t) = \frac{\phi_k(t)}{N_l},\quad
    \lambda = N_l.
\end{equation}

The first and second derivatives of the phase are
\begin{equation}
    f_k'(t) = \frac{2\pi (p t + q_k)}{N_l}, 
    \qquad
    f_k''(t) = \frac{2\pi p}{N_l},
\end{equation}
so the stationary point is
\begin{equation}
    f_k'(n_0(k)) = 0
    \quad\Rightarrow\quad
    n_0(k) = -\frac{q_k}{p}.
    \label{eq:n0_of_k_short}
\end{equation}
The stationary-phase contribution is relevant only when the stationary
point lies inside the observation window:
\begin{equation}
    0 < n_0(k) < N_l.
\end{equation}
The set of indices \(k\) satisfying this inequality forms a contiguous
block \(K \subset \{0,\dots,N_l-1\}\) with cardinality
\begin{equation}
    |K| \approx N_l\Bigl(1 - \frac{N_l}{N_h}\Bigr).
\end{equation}

For \(k \in K\), we can apply the standard stationary-phase formula
(see, e.g., \cite{stationary}) to \eqref{eq:U_integral_approx_short}. Since
\(f_k''(t)\) is constant and independent of \(k\), we obtain the
leading-order approximation
\begin{equation}
    U[k]
    \approx
    \frac{1}{\sqrt{N_l}}\,
    \exp\!\bigl\{j\phi_k(n_0(k))\bigr\}
    \sqrt{\frac{2\pi}{|\phi_k''(n_0(k))|}}\,
    e^{\pm j\pi/4},
\end{equation}
so that the magnitude is approximately
\begin{equation}
    |U[k]|
    \approx
    \frac{1}{\sqrt{N_l|p|}},
    \qquad k \in K,
\end{equation}
and is therefore essentially flat over the contiguous block \(K\). For \(k \notin K\), no stationary point lies in \([0,N_l]\), and the contribution is much smaller.

\bibliographystyle{unsrt} 
\bibliography{ref}        
\vspace{12pt}

\end{document}